\documentclass[12pt,reqno]{amsart}

\newcommand\version{September 22, 2008}

\usepackage{amsmath,amsfonts,amsthm,amssymb,amsxtra}

\newtheorem{theorem}{Theorem}[section]

\newtheorem{lemma}[theorem]{Lemma}

\theoremstyle{definition}

\newtheorem{example}[theorem]{Example}

\theoremstyle{remark}

\newtheorem{remark}[theorem]{Remark}

\numberwithin{equation}{section}

\newcommand{\C}{\mathbb{C}}

\renewcommand{\epsilon}{\varepsilon}

\newcommand{\HLT}{\mathrm{HLT}}
\newcommand{\loc}{{\rm loc}}

\newcommand{\R}{\mathbb{R}}

\newcommand{\Sph}{\mathbb{S}}

\DeclareMathOperator{\tr}{tr}

\begin{document}

\title[Hardy-Lieb-Thirring inequalities --- \version]{A simple proof of Hardy-Lieb-Thirring inequalities}

\author{Rupert L. Frank}
\address{Rupert L. Frank, Department of Mathematics, Princeton University, Princeton, NJ 08544, USA} 
\email{rlfrank@math.princeton.edu}

\begin{abstract}
 We give a short and unified proof of Hardy-Lieb-Thirring inequalities for moments of eigenvalues of fractional Schr\"odinger operators. The proof covers the optimal parameter range. It is based on a recent inequality by Solovej, S\o rensen, and Spitzer. Moreover, we prove that any non-magnetic Lieb-Thirring inequality implies a magnetic Lieb-Thirring inequality (with possibly a larger constant).
\end{abstract}

\thanks{\copyright\, 2008 by the author. This paper may be reproduced, in its entirety, for non-commercial purposes.}

\maketitle


\section{Introduction and main result}

This paper is concerned with estimates on moments of negative eigenvalues of Schr\"odinger operators $ (-\Delta)^s -\mathcal C_{s,d} |x|^{-2s} - V$ in $L_2(\R^d)$ in terms of integrals of the potential $V$. Here 
\begin{equation} \label{eq:csd}
     \mathcal C_{s,d} := 2^{2s} \frac{\Gamma((d+2s)/4)^2}{\Gamma((d-2s)/4)^2}
\end{equation}
is the sharp constant in the Hardy inequality
\begin{equation}\label{eq:hardy}
     \int_{\R^d} |p|^{2s} |\hat u(p)|^2 \,dp
\geq \mathcal C_{s,d} \int_{\R^d} |x|^{-2s} |u(x)|^2 \,dx\,, 
\qquad u\in C_0^\infty(\R^d)\,,
\end{equation}
which is valid for $0<s<d/2$ \cite{He} and we write $\hat u(p) := (2\pi)^{-d/2} \int_{\R^d} u(p) e^{-ip\cdot x}\,dx$ for the Fourier transform of $u$. In \cite{FrLiSe1} we have shown that for any $\gamma>0$, $0<s\leq 1$ and $0<s<d/2$ one has
\begin{equation}\label{eq:hltintro}
\tr\left((-\Delta)^s -\mathcal C_{s,d} |x|^{-2s} - V\right)_-^\gamma
\leq L_{\gamma,d,s}^{\HLT} \int_{\R^d} V(x)_+^{\gamma+d/2s} \,dx
\end{equation}
with a constant $L_{\gamma,d,s}^{\HLT}$ independent of $V$. Here and in the following, $t_\pm:=\max\{\pm t,0\}$ denotes the positve and negative parts of a real number or a self-adjoint operator $t$. The case $s=1$ in \eqref{eq:hltintro} has been shown earlier in \cite{EkFr}. We refer to \eqref{eq:hltintro} as \emph{Hardy-Lieb-Thirring inequality} since it is (up to the value of the constant) an improvement of the Lieb-Thirring inequality \cite{LiTh}
\begin{equation}\label{eq:lt}
\tr\left((-\Delta)^s - V\right)_-^\gamma
\leq L_{\gamma,d,s} \int_{\R^d} V(x)_+^{\gamma+d/2s} \,dx \,.
\end{equation}
It should be pointed out that if $0<s<d/2$, then \eqref{eq:lt} is valid even for $\gamma=0$ (as first shown by Cwikel, Lieb, and Rozenblum) while \eqref{eq:hltintro} is not. We refer to the surveys \cite{LaWe,Hu} for background and references concerning \eqref{eq:lt}.

The original motivation for \eqref{eq:lt} came from the problem of stability of non-relati\-vis\-tic matter \cite{LiSe}. Likewise, our motivation for \eqref{eq:hltintro} was stability of \emph{relativistic} matter in \emph{magnetic fields}. For this problem it is crucial that \eqref{eq:hltintro} continues to holds if $(-\Delta)^s$ is replaced by $(D-A)^{2s}$ with a magnetic vector potential $A\in L_{2,\loc}(\R^d,\R^d)$, and that the constant can be chosen independently of $A$. Here, as usual, $D=-i\nabla$ and the operator $(D-A)^{2s}:=((D-A)^2)^s$ is defined using the spectral theorem. Using the magnetic version of \eqref{eq:hltintro} we could prove stability of relativistic matter in magnetic fields up to and including the critical value of the nuclear charge $\alpha Z=2/\pi=\mathcal C_{1/2,3}$; see \cite{FrLiSe1} and also \cite{FrLiSe2}.

The purpose of this paper is fourfold.
\begin{enumerate}
 \item  We will give a new, much simpler proof of \eqref{eq:hltintro}. While the method in \cite{FrLiSe1} relied on rather involved relations between Sobolev inequalities and decay estimates on heat kernels, the present proof uses nothing more than \eqref{eq:lt} (with $\gamma=0$ and with $s$ replaced by some $t<s$) and the generalization of a powerful (though elementary to prove) new inequality by Solovej, S\o rensen and Spitzer \cite{SoSoSp}.
\item We will extend \eqref{eq:hltintro} to its optimal parameter range $0<s<d/2$. For $d\geq 3$ and $1<s<d/2$ this is a new result, even for integer values of $s$ when the operator is local. This result can not be attained with the method of \cite{FrLiSe1}, since positivity properties of the heat kernel break down for $s>1$.
\item Though our new proof of \eqref{eq:hltintro} does \emph{not} work in the presence of a magnetic field, we shall prove a new operator-theoretic result, which says that any non-magnetic Lieb-Thirring inequality implies a magnetic Lieb-Thirring inequality (with possibly a different constant). This recovers, in particular, that \eqref{eq:hltintro} holds if $(-\Delta)^s$ is replaced by $(D-A)^{2s}$ and $0<s\leq 1$. (The reason for the restriction $s\leq 1$ at this point is that we need a diamagnetic inequality.) Another application of this result concerns the recent inequality in \cite{KoVuWe} corresponding to the endpoint $\gamma=0$ of \eqref{eq:lt} with $s=1$, $d=2$ .
\item We show that an analog of inequality \eqref{eq:hltintro} for $s=1/2$, $d=3$ holds in a model for pseudo-relativistic electrons that includes spin. The difficulty here is that the potential energy is non-local. This new estimate simplifies some of the proofs in \cite{FrSiWa} and will be, we believe, a crucial ingredient in the proof of stability of matter in this model.
\end{enumerate}

Here is the precise statement of our result.

\begin{theorem}\label{hlt}
 Let $d\geq 1$, $0<s<d/2$ and $\gamma>0$. Then there is a constant $L_{\gamma,d,s}^\HLT$ such that
\begin{equation}\label{eq:hlt}
 \tr\left((-\Delta)^s -\mathcal C_{s,d} |x|^{-2s} - V\right)_-^\gamma
\leq L_{\gamma,d,s}^\HLT \int_{\R^d} V(x)_+^{\gamma+d/2s} \,dx \,.
\end{equation}
If $d\geq 2$, $0<s\leq 1$ and $(-\Delta)^s$ is replaced by $(D-A)^{2s}$ for some $A\in L_{2,\loc}(\R^d,\R^d)$, then \eqref{eq:hlt} remains valid if $L_{\gamma,d,s}^\HLT$ is replaced by $L_{\gamma,d,s}^\HLT \,(e/p)^p \,\Gamma(p+1)$ with $p=\gamma+d/2s$.
\end{theorem}

The crucial ingredient in our proof of \eqref{eq:hlt} is the following lower bound for the quadratic form
$$
h_s[u] := \int_{\R^d} |p|^{2s} |\hat u(p)|^2 \,dp - \mathcal C_{s,d} \int_{\R^d} |x|^{-2s} |u(x)|^2 \,dx
$$
of the operator $ (-\Delta)^s -\mathcal C_{s,d} |x|^{-2s}$.

\begin{theorem}\label{hardyrem}
 Let $0<t<s<d/2$. Then there exists a constant $\kappa_{d,s,t}>0$ such that for all $u\in C_0^\infty(\R^d)$ one has
\begin{equation}\label{eq:hardyremscal}
h_s[u]^\theta \|u\|^{2(1-\theta)} \geq \kappa_{d,s,t} \|(-\Delta)^{t/2} u\|^2 \,,
\qquad \theta:=t/s\,.
\end{equation}
\end{theorem}

In the special case $d=3$ and $s=1/2$ this is a recent result by Solovej, S\o rensen and Spitzer \cite[Thm. 11]{SoSoSp}. The results reported here are motivated by their work. Below we shall show that their proof extends to arbitrary $0<s<d/2$.

Our original proof of \eqref{eq:hlt} in \cite{FrLiSe1} for $0<s\leq 1$ relied on the Gagliardo-Nirenberg-type inequality
\begin{equation}\label{eq:hs}
h_s[u]^\theta \|u\|^{2(1-\theta)} \geq \sigma_{d,s,q} \|u\|_q^2 \,,
\qquad \theta:=\frac ds\left(\frac12-\frac1q\right) \,,
\end{equation}
for $2<q<2s/(d-2s)$. This is weaker than \eqref{eq:hardyremscal} in view of the Sobolev inequality \cite[Thms. 4.3 and 8.3]{LiLo}
$$
\|(-\Delta)^{t/2} u\|^2 \geq S_{d,t} \|u\|_q^2 \,,
\qquad q=\frac{2d}{d-2t} \,.
$$
What makes \eqref{eq:hardyremscal} much easier to prove than \eqref{eq:hs} is that it is a \emph{linear} inequality, that is, all norms are taken in $L_2(\R^d)$. Indeed, \eqref{eq:hardyremscal} is easily seen to be equivalent to the operator inequality
\begin{equation}\label{eq:hardyrem}
(-\Delta)^{s} - \mathcal C_{s,d} |x|^{-2s} \geq K_{d,s,t} l^{-2(s-t)}(-\Delta)^t - l^{-2s} \,,
\qquad l>0\,,
\end{equation}
where $K_{d,s,t}= \left( s^{-s} t^t (s-t)^{s-t} \right)^{1/s} \kappa_{d,s,t}$, and this is the way we shall prove it in the next section.

\textbf{Acknowledgements.} The author would like to thank E. Lieb and R. Seiringer for very fruitful discussions, as well as J. P. Solovej, T. {\O}stergaard S{\o}rensen and W.~Spitzer for useful correspondence. Support through DAAD grant D/06/49117 and U.S. National Science Foundation grant PHY 06 52854 is gratefully acknowledged.


\section{Proof of Theorem \ref{hardyrem}}

Throughout this section we assume that $0<s<d/2$. Recall that for $0<\alpha<d$ the Fourier transform of $|x|^{-d+\alpha}$ is given by
\begin{equation}\label{ft1x}
     b_{d-\alpha} \left(|\cdot|^{-d+\alpha}\right)^\wedge (p)
     = b_{\alpha} |p|^{-\alpha},
     \qquad b_\alpha := 2^{\alpha/2} \Gamma(\alpha/2)\,;
\end{equation}
see, e.g., \cite[Thm.~5.9]{LiLo}, where another convention for the Fourier transform is used, however. This implies that for $2s<\alpha<d$ one has
\begin{equation}\label{eq:convol}
\int_{\R^d} \frac1{|p-q|^{d-2s} |q|^{\alpha}} \,dq = \Psi_{s,d}(\alpha) \frac1{|p|^{\alpha-2s}} \,,
\end{equation}
where
\begin{equation*}\label{eq:psi}
 \Psi_{s,d}(\alpha)
:= (2\pi)^{d/2} \frac{b_{2s} \, b_{\alpha-2s} \, b_{d-\alpha}}{b_{d-2s} \, b_{d-\alpha+2s} \,  b_{\alpha}}
= \frac{\pi^{d/2} \,\Gamma(s)}{\Gamma((d-2s)/2)} \ \frac{\Gamma((\alpha-2s)/2)\,\Gamma((d-\alpha)/2)}{\Gamma((d-\alpha+2s)/2)\,\Gamma(\alpha/2)}
 \,.
\end{equation*}
We shall need the following facts about $\Psi_{s,d}(\alpha)$ as a function of $\alpha\in(2s,d)$.

\begin{lemma}\label{incr}
 $\Psi_{s,d}$ is an even function with respect to $\alpha=(d+2s)/2$ and one has
\begin{equation}\label{eq:psimin}
\Psi_{s,d}((d+2s)/2)= (2\pi)^{d/2} \frac{b_{2s}}{b_{d-2s}}\mathcal C_{s,d}^{-1}
\end{equation}
with $\mathcal C_{s,d}$ from \eqref{eq:csd}. Moreover, $\Psi_{s,d}$ is strictly decreasing on $(2s, (d+2s)/2)$ and strictly increasing on $((d+2s)/2,d)$.
\end{lemma}

This is Lemma 3.2 from \cite{FrLiSe1} in disguise.

\begin{proof}[Proof of Lemma \ref{incr}]
$\Psi_{s,d}(\alpha)$ is obviously invariant under replacing $\alpha$ by $d+2s-\alpha$, and its value at $\alpha=(d+2s)/2$ follows immediately from definition \eqref{eq:csd}. To prove the monotonicity we write
\begin{equation*}
\Psi_{s,d}(\alpha) =  \frac{\pi^{d/2}\ \Gamma(s)}{\Gamma((d-2s)/2)} \ \frac{f(t)}{f(s+t)}\,,
\qquad t=(\alpha-2s)/2 \,,
\end{equation*}
where $T:=(d-2s)/2$ and $f(t):= \Gamma(t)/\Gamma(T+s-t)$. We need to show that $\log(f(t)/f(s+t))$ is strictly decreasing in $t\in (0,T/2)$. Noting that
$$
\frac{f'(t)}{f(t)} = \psi(t) + \psi(T+s-t)
$$
with $\psi:=\Gamma'/\Gamma$ the Digamma function, we have
$$
\frac{d}{dt} \log \frac{f(t)}{f(t+s)} = \psi(t) + \psi(T+s-t) - \psi(t+s) - \psi(T-t)
=-\int_{t}^{t+s}  h(\tau) \,d\tau
$$
with $h(\tau):= \psi'(\tau)-\psi'(T+s-\tau)$ for $0<\tau<T+s$. Since $\psi'$ is strictly decreasing \cite[(6.4.1)]{AbSt}, $h$ is an odd function with respect to $\tau=(T+s)/2$ which is strictly positive for $\tau<(T+s)/2$. Since the midpoint of the interval $(t,t+s)$ lies to the left of $(T+s)/2$, the integral of $h$ over this interval is strictly positive, which proves the claim.
\end{proof}

Now we prove \eqref{eq:hardyrem}, following the strategy of Solovej, S\o rensen and Spitzer \cite{SoSoSp} in the special case $d=3$, $s=1/2$; see also \cite[Thm. 11]{LiYa} for a related argument.

\begin{proof}[Proof of Theorem \ref{hardyrem}]
For technical reasons we prove the theorem only for $2s/3\leq t<s$. It is easy to see that this implies the result for all $0<t<s$.

By a well-known argument (going back at least to Abel and, in the present context, to \cite{KoPeSe}) based on the Cauchy-Schwarz inequality one has for any positive measurable function $h$ on $\R^d$
\begin{equation*}
 (2\pi)^{d/2} \frac{b_{2s}}{b_{d-2s}} \int_{\R^d} \frac{|u|^2}{|x|^{2s}} \,dx 
= \iint_{\R^d\times\R^d} \frac{\overline{\hat u(p)} \hat u(q)}{|p-q|^{d-2s}} \,dp\,dq
\leq \int_{\R^d} t_h(p) |\hat u(p)|^2 \,dp \,,
\end{equation*}
where
$$
t_h(p):= h(p)^{-1} \int_{\R^d} \frac{h(q)}{|p-q|^{d-2s}} \,dq \,.
$$
Below we shall choose $h$ (depending on $l>0$) in such a way that for some positive constants $A$ and $B$ (depending on $d$, $s$ and $t$, but not on $l$) one has
\begin{equation}
 \label{eq:tbound}
t_h(p) \leq \Psi_{s,d}((d+2s)/2) |p|^{2s} - A l^{-2(s-t)} |p|^{2t} + B l^{-2s} \,.
\end{equation}
(By scaling it would be enough to prove this for $l=1$, but we prefer to keep $l$ free.) Because of \eqref{eq:psimin} this estimate proves \eqref{eq:hardyrem}.

We show that \eqref{eq:tbound} holds with $h(p)=(|p|^{(d+2s)/2} + l^{\beta-(d+2s)/2} |p|^\beta)^{-1}$ where $\beta$ is a parameter depending on $t$ that will be fixed later. (Indeed, we shall choose $\beta=2t+(d-2s)/2$.) Since the derivatives of the function $r\mapsto r^{-1}$ have alternating signs one has $(a+b)^{-1} \leq a^{-1} - a^{-2} b + a^{-3} b^2$ and therefore
$$
\int_{\R^d} \frac{h(q)}{|p-q|^{d-2s}} \,dq 
\leq \int_{\R^d} \frac{1}{|p-q|^{d-2s}}\left(\frac1{|q|^{(d+2s)/2}} - \frac{l^{\beta-(d+2s)/2}}{|q|^{d+2s-\beta}} + \frac{l^{2\beta-d-2s}}{|q|^{3(d+2s)/2-2\beta}} \right) \,dq \,.
$$
If we assume that $(d+6s)/4<\beta<(3d+2s)/4$ then the right side is finite and, using notation \eqref{eq:convol} with $\Psi$ instead of $\Psi_{s,d}$, equal to
$$
\Psi\left(\frac{d+2s}2\right) \frac1{|p|^{(d-2s)/2}} - \Psi(d+2s-\beta) \frac{l^{\beta-(d+2s)/2}}{|p|^{d-\beta}}
 + \Psi\left(\frac{3(d+2s)}2-2\beta\right) \frac{l^{2\beta-d-2s}}{|p|^{3d/2-2\beta+s}} \,.
$$
Thus
\begin{align*}
t_h(p) \leq &
\Psi\left(\frac{d+2s}2\right) |p|^{2s} 
- \left(\Psi(d+2s-\beta) - \Psi\left(\frac{d+2s}2\right) \right) l^{\beta-(d+2s)/2} |p|^{\beta-(d-2s)/2} \\
& + \left(\Psi\left(\frac{3(d+2s)}2-2\beta\right) - \Psi(d+2s-\beta) \right) l^{2\beta-d-2s} |p|^{2\beta-d} \\
& + \Psi\left(\frac{3(d+2s)}2-2\beta\right) l^{3\beta-3d/2-3s} |p|^{3\beta-3d/2-s} \,.
\end{align*}
If we assume that $\beta\leq (d+2s)/2$, then the exponents of $|p|$ on the right side satisfy $2s\geq\beta-(d-2s)/2 \geq 2\beta-d \geq 3\beta-3d/2-s$, and if $\beta\geq(3d+2s)/6$ then the last exponent is non-negative. Now we choose $\beta=2t+(d-2s)/2$, so that the exponent of the second term is $2t$ and the condition $\beta\geq(3d+2s)/6$ is satisfied, since we are assuming that $t\geq 2s/3$. Moreover, according to Lemma \ref{incr}, the coefficient of the second term is negative. Finally, we use that there are constants $C_1$ and $C_2$ such that for any $\epsilon>0$ one has
$$
|p|^{2\beta-d} \leq \epsilon |p|^{\beta-(d-2s)/2} + C_1 \epsilon^{-\frac{2(2\beta+d)}{d+2s-2\beta}} \,,
\quad
|p|^{3\beta-3d/2-s} \leq \epsilon |p|^{\beta-(d-2s)/2} + C_2 \epsilon^{-\frac{6\beta-3d-2s}{2(d+2s-2\beta)}} \,.
$$
This concludes the proof of \eqref{eq:tbound}.
\end{proof}


\section{Proof of Theorem \ref{hlt}}

We fix $0<s<d/2$ and $\gamma>0$ and write
$$
\tr\left((-\Delta)^s -\mathcal C_{s,d}|x|^{-2s} -V\right)_-^\gamma
=\gamma \int_0^\infty N(-\tau, (-\Delta)^s -\mathcal C_{s,d}|x|^{-2s} -V) \, \tau^{\gamma-1}\,d\tau \,,
$$
where $N(-\tau,H)$ denotes the number of eigenvalues less than $-\tau$, counting multiplicities, of a self-adjoint operator $H$. We shall use \eqref{eq:hardyrem} with $l^{-2s}=\sigma\tau$ and some $0<t<s$ and $0<\sigma<1$ to be specified below. Abbreviating $K_t=K_{d,s,t}$ we find that
\begin{align*}
N(-\tau, (-\Delta)^s -\mathcal C_{s,d}|x|^{-2s} -V) & \leq N(0, K_t (\sigma\tau)^{(s-t)/s}(-\Delta)^t -V + (1-\sigma)\tau) \\
& = N\left(0, (-\Delta)^s - K_t^{-1} (\sigma\tau)^{-(s-t)/s} \left(V - \left(1-\sigma\right)\tau\right) \right) \,.
\end{align*}
Now we use \eqref{eq:lt} with $\gamma=0$ and $s$ replaced by $t$ (see \cite{Da} for $t\leq 1$ and \cite{Cw} for $t<d/2$). Abbreviating $L_t=L_{0,d,t}$ we have
$$
N(-\tau, (-\Delta)^s -\mathcal C_{s,d}|x|^{-2s} -V)
\leq L_{t} K_t^{-d/2t} (\sigma\tau)^{-d(s-t)/2st} \int_{\R^d} \left(V - \left(1-\sigma\right)\tau\right)_+^{d/2t} \,dx
$$
and
\begin{align*}
& \tr\left( (-\Delta)^s -\mathcal C_{s,d}|x|^{-2s} -V\right)_-^\gamma \\
& \quad \leq \gamma L_{t} K_t^{-d/2t} \sigma^{-d(s-t)/2st} \int_{\R^d} dx \int_0^\infty d\tau \tau^{\gamma-1-d(s-t)/2st} \left(V - \left(1-\sigma\right)\tau\right)_+^{d/2t} \,dx \\
& \quad = \gamma L_{t} K_t^{-d/2t} \sigma^{-\frac{d(s-t)}{2st}} (1-\sigma)^{-\gamma+\frac{d(s-t)}{2st}} \ \frac{\Gamma(\gamma-\tfrac{d(s-t)}{2st}) \Gamma(\tfrac d{2t}+1)}{\Gamma(\gamma+\tfrac d{2s}+1)} \
 \int_{\R^d} V_+^{\gamma+d/2s} \,dx \,.
\end{align*}
Here we assumed that $t>ds/(2\gamma s+d)$ so that the $\tau$ integral is finite. Finally, we optimize over $0<\sigma<1$ by choosing $\sigma=d(s-t)/2\gamma st$ and over $ds/(2\gamma s+d)<t<s$ to complete the proof of \eqref{eq:hlt}.

The statement about the inclusion of $A$ follows from Example \ref{diamagex} and Theorem \ref{diamagneg} in the following section.


\section{Magnetic Lieb-Thirring inequalities}

In this section we discuss Lieb-Thirring inequalities for magnetic Schr\"odinger operators, that is, \eqref{eq:lt} (and its generalizations) with $(-\Delta)^s$ replaced by $(D-A)^{2s}$ for some vector field $A\in L_{2,\loc}(\R^d,\R^d)$. 

It is a remarkable fact that all presently known proofs of Lieb-Thirring inequalities, which allow for the inclusion of a magnetic field, yield the same constants in the magnetic case as in the non-magnetic case. It is unknown whether this is also true for the unknown sharp constants. Note that the diamagnetic inequality implies that the lowest eigenvalue does not decrease when a magnetic field is added, but there is no such result for, e.g., the number or the sum of eigenvalues; see \cite{AvHeSi, Li2}. Rozenblum \cite{Ro} discovered, however, that any power-like bound on the number of eigenvalues in the non-magnetic case implies a similar bound in the magnetic case, with possibly a worse constant. Here we show the same phenomenon for \emph{moments} of eigenvalues.

We work in the following abstract setting. Let $(X,\mu)$ be a sigma-finite measure space and let $H$ and $M$ be non-negative operators in $L_2(X,\mu)$ such that for any $u\in L_2(X,\mu)$ and any $t>0$
\begin{equation}\label{eq:domination}
|\exp(-tM) u(x)| \leq (\exp(-tH)|u|)(x)
\qquad \mu-\text{a.e.}\ x\in X \,.
\end{equation}
Note that this implies that $\exp(-tH)$ is positivity preserving. We think of $H$ as a non-magnetic operator, $M$ a magnetic operator and \eqref{eq:domination} as a diamagnetic inequality. It might be useful to keep the following example in mind.

\begin{example}\label{diamagex}
 Let $X=\R^d$ with Lebesgue measure, $H=(-\Delta)^s$, and $M=(D-A)^{2s}$ for some $0<s\leq 1$ and $A\in L_{2,\loc}(\R^d)$. The diamagnetic inequality \eqref{eq:domination} in the case $s=1$ was shown in \cite{Si1}, and in the case $0<s<1$ it follows from the $s=1$ case since the function $\lambda\mapsto\exp(-\lambda^s)$ is completely monotone and hence by Bernstein's theorem \cite{Do} the Laplace transform of a positive measure. More generally, \eqref{eq:domination} holds for $H=(-\Delta)^s+W$ and $M=(D-A)^{2s}+W$ with $s$ and $A$ as before and a, say, bounded function $W$. This can be seen using Trotter's product formula. By an approximation argument the inequality holds also for $W(x)=-\mathcal C_{s,d}|x|^{-2s}$.
\end{example}

The main result in this section is

\begin{theorem}\label{diamagneg}
Let $H$ and $M$ be as above and assume that there exist some constants $L>0$, $\gamma\geq 0$, $p>0$ and a non-negative function $w$ on $X$ such that for all $V\in L_p(V,w\,d\mu)$ one has
\begin{equation}\label{eq:diamagnegass}
\tr(H-V)_-^{\gamma} \leq L \int_X V_+^{p} w \,d\mu \,.
\end{equation}
Then one also has
\begin{equation}\label{eq:diamagneg}
\tr(M-V)_-^\gamma \leq L \left(\frac ep\right)^p \Gamma(p+1) \int_X V_+^p w \,d\mu \,.
\end{equation}
\end{theorem}

We do not know whether the factor $(e/p)^p \Gamma(p+1)$ in \eqref{eq:diamagneg} can be omitted. Results from \cite{FrLoWe} about the eigenvalues of the Landau Hamiltonian in a domain (but without potential) seem to indicate that a factor $>1$ is necessary. Our proof of Theorem \ref{diamagneg} uses some ideas from \cite{Ro} where the case $\gamma=0$ was treated; see also \cite{FrLiSe2} for a result about operators with discrete spectrum.

\begin{remark}
 With the same proof one can deduce estimates on $\tr f(M)$ from estimates on $\tr f(H)$ for more general functions $f$. For example, let $d=2$ and $f(t):=|\ln |t||^{-1}$ if $- e^{-1}< t<0$, $f(t):=1$ if $t\leq -e^{-1}$, and $f(t):=0$ if $t\geq 0$. Then there exists a constant $L$ and for any $q>1$ a constant $L_q$ such that for all $l>0$ and $A\in L_{2,\loc}(\R^2,\R^2)$
$$
\tr f\left(l^2((D-A)^2-V)\right)
\leq L \int_{|x|<l} \! V(x)_+ \left|\log\frac{|x|}l\right| \,dx + L_q \int_0^\infty \!\!\left(\int_{\Sph} V(r\omega)_+^q \,d\omega \right)^{1/q} \!r\,dr \,.
$$
Indeed, this follows by Lemma \ref{average} via integration from the $A\equiv 0$ result of \cite{KoVuWe}.
\end{remark}

The key ingredient in the proof of Theorem \ref{diamagneg} is a bound on the negative eigenvalues of $M-V$ by those of $H-\alpha V$, averaged over all coupling constants $\alpha$. As before, we denote by $N(-\tau,A)$ the number of eigenvalues less than $-\tau$, counting multiplicities, of a self-adjoint operator $A$.

\begin{lemma}\label{average}
Let $H$ and $M$ be non-negative self-adjoint operators satisfying \eqref{eq:domination} and let $V\geq 0$. Then for any $\tau\geq 0$ and $t>0$ one has 
\begin{equation}\label{eq:diamagnegproof}
N(-\tau,M- V) \leq t e^t \int_0^\infty N(-\tau,H-\alpha V) e^{-\alpha t}\,d\alpha \,.
\end{equation}
\end{lemma}

\begin{proof}
Since \eqref{eq:domination} remains valid with $H+\tau$ and $M+\tau$ in place of $H$ and $M$ we need only consider $\tau=0$. Moreover, by a density argument we may assume that $V>0$ a.e. We define $h:=V^{-1/2} H V^{-1/2}$ and $m:=V^{-1/2} M V^{-1/2}$ via quadratic forms and claim that \eqref{eq:domination} holds with $h$ and $m$ in place of $H$ and $M$. Since this fact is proved in \cite[Thm. 3]{Ro} we only sketch the main idea. Indeed, for any $\sigma>0$
\begin{equation*}
(m+\sigma)^{-1} = V^{1/2} (M+\sigma V)^{-1} V^{1/2} 
= \int_0^\infty V^{1/2}\exp(-s(M+\sigma V)) V^{1/2} \,ds\,,
\end{equation*}
and by \eqref{eq:domination} and Trotter's product formula $|\exp(-s(M+\sigma V)) V^{1/2} u| \leq \exp(-s(H+\sigma V)) V^{1/2} |u|$ a.e. Hence $|(m+\sigma)^{-1}u| \leq (h+\sigma)^{-1} |u|$ a.e. Iterating this inequality and recalling that $(1+tm/n)^{-n} \to \exp(-tm)$ strongly as $n\to\infty$, we obtain \eqref{eq:domination} for $h$ and $m$.

By \cite[Thm. 4.1]{Si} this analog of \eqref{eq:domination} implies that
$$
\tr\exp(-tm) = \|\exp(-tm/2)\|_2^2\leq \|\exp(-th/2)\|_2^2 = \tr\exp(-th)
$$
with $\|\cdot\|_2$ the Hilbert-Schmidt norm, and hence by the Birman-Schwinger principle
\begin{equation*}
N(M-V) = N(1,m)\leq e^t \tr\exp(-tm) \leq e^t \tr\exp(-th) \,.
\end{equation*}
Using the Birman-Schwinger principle once more, we find
\begin{equation*}
\tr\exp(-th) = t \int_0^\infty N(\alpha,h) e^{-t\alpha}\,d\alpha
= t \int_0^\infty N(H-\alpha V) e^{-t\alpha}\,d\alpha \,,
\end{equation*}
proving \eqref{eq:diamagnegproof}.
\end{proof}

\begin{proof}[Proof of Theorem \ref{diamagneg}]
By the variational principle we may assume that $V\geq 0$. By Lemma \ref{average} one has for any $t>0$
\begin{align*}
\tr(M-V)_-^\gamma
& = \gamma \int_0^\infty N(-\tau, M-V) \tau^{\gamma-1} \,d\tau \\ 
& \leq \gamma t e^t \int_0^\infty \int_0^\infty N(-\tau,H-\alpha V) \tau^{\gamma-1} \,d\tau e^{-\alpha t}\,d\alpha \\
& = t e^t \int_0^\infty \tr(H-\alpha V)_-^\gamma e^{-\alpha t}\,d\alpha \,,
\end{align*}
and by assumption \eqref{eq:diamagnegass} the right hand side can be bounded from above by
\begin{equation*}
L t e^t \left(\int_0^\infty \alpha^p e^{-\alpha t}\,d\alpha \right) \int_X V^p w \,d\mu
= L t^{-p} e^t \Gamma(p+1) \int_X V^p w \,d\mu \,.
\end{equation*}
Now the assertion follows by choosing $t=p$.
\end{proof}


\section{A pseudo-relativistic model including spin}

Throughout this section we assume that $d=3$. The helicity operator $h$ on $L_2(\R^3,\C^2)$ is defined as the Fourier multiplier corresponding to the matrix-valued function $p\mapsto \mathbf\sigma\cdot p/|p|$, where $\mathbf\sigma=(\sigma_1,\sigma_2,\sigma_3)$ denotes the triple of Pauli matrices. The properties of these matrices imply that $h$ is a unitary and self-adjoint involution. The analog of the Hardy (or Kato) inequality \eqref{eq:hardy} is
\begin{equation}\label{eq:eps}
     \int_{\R^3} |\xi| |\hat u(\xi)|^2 \,d\xi
\geq \tilde{\mathcal C} \int_{\R^3} \frac{|u(x)|^2 + |(hu)(x)|^2}{2 \, |x|} \,dx\,, 
\qquad u\in C_0^\infty(\R^3,\C^2)\,,
\end{equation}
with the sharp constant
$$
\tilde{\mathcal C} = \frac{2}{2/\pi+\pi/2} \,;
$$
see \cite{EvPeSi}. Note that this constant is strictly larger than
$$
\mathcal C:= \mathcal C_{1/2,3}=2/\pi \,,
$$
which is the constant one would get if $hu$ were replaced by $u$ on the right side of \eqref{eq:eps}.

For a function $V$ on $\R^3$ taking values in the Hermitean $4\times4$ matrices we introduce the non-local potential $$
\Phi(V) := \frac12 \begin{pmatrix}1_{L_2(\R^3,\C^2)} \\ h\end{pmatrix}^* V \begin{pmatrix}1_{L_2(\R^3,\C^2)} \\ h\end{pmatrix} \,,
$$
where $\begin{pmatrix}1_{L_2(\R^3,\C^2)} \\ h\end{pmatrix}$ is considered as an operator from $L_2(\R^3,\C^2)$ to $L_2(\R^3,\C^4)$. The operator $\sqrt{-\Delta} - \Phi(V)$ in $L_2(\R^3,\C^2)$ has been suggested by Brown and Ravenhall as the Hamiltonian of a massless, relativistic spin-1/2 particle in a potential $-V$. It results from projecting onto the positive spectral subspace of the Dirac operator. One of the advantages of this operator over the simpler $\sqrt{-\Delta} - V$ is that it is well-defined for nuclear charges $\alpha Z\leq \tilde{\mathcal C}$, which includes all known elements. We refer to \cite{LiSe} for more background about this model. Despite the efforts in \cite{LiSiSo,BaEv,HoSi} the problem of stability of matter for the corresponding many-particle system is not yet completely understood and the following result, we believe, might be useful in this respect.

\begin{theorem}\label{hltbr}
 Let $d=3$ and $\gamma>0$. Then there is a constant $\tilde L_{\gamma}^\HLT$ such that
\begin{equation}\label{eq:hltbr}
 \tr\left(\sqrt{-\Delta} - \tilde{\mathcal C} \Phi(|x|^{-1})-\Phi(V) \right)_-^\gamma
\leq \tilde L_{\gamma}^\HLT \int_{\R^3} \tr_{\C^4} V(x)_+^{\gamma+3} \,dx \,.
\end{equation}
\end{theorem}

For the proof of this theorem we need some facts about the partial wave decomposition of the operator $\sqrt{-\Delta} - \tilde{\mathcal C} \Phi(|x|^{-1})$ from \cite{EvPeSi}. This operator commutes with the total angular momentum operator $\mathbf J=\mathbf L+\frac12{\bf\sigma}$, where $\mathbf L=-i\nabla\times x$, as well as with the operator $\mathbf L^2$. The subspace corresponding to total angular momentum $j=1/2$ is of the form $\mathfrak H_{1/2,0} \oplus \mathfrak H_{1/2,1}$, where the subspaces $\mathfrak H_{1/2,l}$ correspond to the eigenvalues $l(l+1)$ of $\mathbf L^2$.

The next result, essentially contained in \cite{FrSiWa}, says that on the space $\mathfrak H_{1/2,0} \oplus \mathfrak H_{1/2,1}$ the operator $\sqrt{-\Delta} - \tilde{\mathcal C} \Phi(|x|^{-1})$ is controlled by the operator $\sqrt{-\Delta} - \mathcal C |x|^{-1}$ with the \emph{smaller} coupling constant $\mathcal C$. (Strictly speaking, the latter operator should be tensored with $1_{\C^2}$, but we suppress this if there is no danger of confusion.)

\begin{lemma}\label{comp}
 If $0\not\equiv\psi\in\mathfrak H_{1/2,0}\cap C_0^\infty(\R^3,\C^2)$, then
\begin{equation*}\label{eq:comp1}
 \frac2{1+(2/\pi)^2} \geq 
\frac{\left(\psi, \left(\sqrt{-\Delta} - \tilde{\mathcal C} \Phi(|x|^{-1})\right)\psi \right)}{\left(\psi, \left(\sqrt{-\Delta} - \mathcal C |x|^{-1}\right)\psi \right)}
\geq \frac1{1+(2/\pi)^2} \,.
\end{equation*}
If $0\not\equiv\psi\in\mathfrak H_{1/2,1}\cap C_0^\infty(\R^3,\C^2)$, this bound is true provided $\left(\psi, \left(\sqrt{-\Delta} - \mathcal C |x|^{-1}\right)\psi \right)$ is replaced by $\left(h\psi, \left(\sqrt{-\Delta} - \mathcal C |x|^{-1}\right)h\psi \right)$.
\end{lemma}

\begin{proof}[Proof of Lemma \ref{comp}]
We prove the assertion only for $l=1$ since the lower bound for $l=0$ is contained in \cite[Lemma 2.7]{FrSiWa} and the upper bound is proved as below. By orthogonality we may assume that the Fourier transform of $\psi$ is of the form $\hat\psi(\xi) = |\xi|^{-2} g(|\xi|) \Omega_{1/2,1,m}(\frac{\xi}{|\xi|})$ where $m\in\{-1/2,1/2\}$ and $\Omega_{1/2,1,m}$ are explicit functions in $L_2(\Sph^2,\C^2)$. By the properties of these functions one has $\widehat{h\psi}(\xi) = - |\xi|^{-2} g(|\xi|) \Omega_{1/2,0,m}(\frac{\xi}{|\xi|})$. According to the ground state representation \cite[Lem\-ma 2.6]{FrSiWa} one has
\begin{align*}
\left(\psi, \left(\sqrt{-\Delta} - \tilde{\mathcal C} \Phi(|x|^{-1})\right)\psi \right)
& = \frac{\tilde{\mathcal C}}{2\pi} \int_0^\infty \int_0^\infty |g(p)-g(q)|^2 \tilde k(\tfrac12(\tfrac pq +\tfrac qp)) \frac{dp}p \, \frac{dq}{q} \,, \\
\left(h\psi, \left(\sqrt{-\Delta} - \mathcal C |x|^{-1}\right)h\psi \right)
& = \frac{\mathcal C}{2\pi} 
\int_0^\infty \int_0^\infty |g(p)-g(q)|^2 k(\tfrac12(\tfrac pq +\tfrac qp)) \frac{dp}p \, \frac{dq}{q} \,,
\end{align*}
where $\tilde k(t)= \frac12 (Q_0(t)+Q_1(t))$, $k(t)=Q_0(t)$, and $Q_l$ are the Legendre functions of the second kind \cite[8.4]{AbSt}. The assertion now follows from the fact that $Q_0\geq Q_1\geq 0$.
\end{proof}

\begin{proof}[Proof of Theorem \ref{hltbr}]
We first claim that for any $0<t<1/2$ there is a $\tilde K_t>0$ such that
\begin{equation}\label{eq:hardyrembr}
 \sqrt{-\Delta} - \tilde{\mathcal C} \Phi(|x|^{-1})
\geq \tilde K_t l^{-1+2t} (-\Delta)^t - l^{-1} \,,
\quad l>0\,.
\end{equation}
Indeed, it follows from Lemma \ref{comp} and \eqref{eq:hardyrem} that on $\mathfrak H_{1/2,0}\oplus\mathfrak H_{1/2,1}$ one has for any $0<t<1/2$
$$
\sqrt{-\Delta} - \tilde{\mathcal C} \Phi(|x|^{-1}) 
\geq \left(1+(2/\pi)^2\right)^{-1} \left( K_t l^{-1+2t} (-\Delta)^t - l^{-1} \right) \,,
\quad l>0 \,.
$$
On the other hand, the arguments of \cite{EvPeSi} show that there exists a constant $\tilde{\mathcal C}'>\tilde{\mathcal C}$ such that $\sqrt{-\Delta} \geq \tilde{\mathcal C}'\Phi(|x|^{-1})$ on 
$\left(\mathfrak H_{1/2,0}\oplus\mathfrak H_{1/2,1}\right)^\bot$. Hence on that space
$$
\sqrt{-\Delta} - \tilde{\mathcal C} \Phi(|x|^{-1}) 
\geq \frac{\tilde{\mathcal C}'-\tilde{\mathcal C}}{\tilde{\mathcal C}'} \sqrt{-\Delta}
\geq \frac{\tilde{\mathcal C}'-\tilde{\mathcal C}}{\tilde{\mathcal C}'} 
\left( \frac 1{2t} l^{-1+2t} (-\Delta)^t -  \frac{1-2t}{2t} l^{-1} \right) \,,
\quad l>0 \,.
$$
This proves \eqref{eq:hardyrembr}.

Given \eqref{eq:hardyrembr}, the proof of \eqref{eq:hltbr} is similar to that of \eqref{eq:hlt}. We may assume that $V(x)=v(x) I_{\C^4}$ for a non-negative, \emph{scalar} function $v$ (otherwise, replace $V(x)$ by $v(x) I_{\C^4}$ where $v(x)$ is the operator norm of the $4\times 4$ matrix $V(x)_+$). For a given $l>0$ and $0<t<1/2$ we introduce the operator $H:= \tilde K_t l^{-1+2t} (-\Delta)^t - v -l^{-1}$ in $L_2(\R^3,\C)$. Then according to \eqref{eq:hardyrembr} one has
\begin{equation*}\label{eq:scalar}
N(-\tau, \sqrt{-\Delta} - \tilde{\mathcal C} \Phi(|x|^{-1}) -\Phi(V))
\leq N(-\tau, \tfrac12 (H\otimes 1_{\C^2} +h (H\otimes 1_{\C^2}) h))
\leq 4 N(-\tau, H) \,.
\end{equation*}
In the last inequality we used that $N(-\tau, \frac12(A+B))\leq N(-\tau,A) + N(-\tau,B)$ for any self-adjoint, lower semi-bounded operators $A$ and $B$, which follows from the variational principle. Now one can proceed in the same way as in the proof of \eqref{eq:hlt}.
\end{proof}


\bibliographystyle{amsalpha}

\end{document}